\documentclass[10pt]{article}
\usepackage{authblk}
\pdfoutput=1

\usepackage[english]{babel}
\usepackage{ucs}
\usepackage[utf8x]{inputenc}

\usepackage{amsfonts}
\usepackage{amsmath}
\usepackage{amssymb}
\usepackage{amsthm}

\usepackage[ruled,linesnumbered]{algorithm2e}
\usepackage{hyperref}

\usepackage{fullpage}

\newcommand{\R}{\mathbb{R}}

\newcommand{\mat}[1]{\mathbf{#1}}
\newcommand{\rank}{\operatorname{rank}}
\newcommand{\diag}{\operatorname{diag}}
\newcommand{\vecm}{\operatorname{vec}}

\newcommand{\norm}[1]{\|#1\|}
\newcommand{\sval}{\mathfrak{s}}

\newcommand{\ltwo}{\ell^2}
\newcommand{\lone}{\ell^1}
\newcommand{\linf}{\ell^{\infty}}
\newcommand{\lp}{\ell^p}
\newcommand{\lpq}{\ell^q}
\newcommand{\ltworat}{\ell^2_{\cR}}
\newcommand{\lonerat}{\ell^1_{\cR}}
\newcommand{\lprat}{\ell^p_{\cR}}

\newcommand{\sstar}{\Sigma^\star}

\newcommand{\cR}{\mathcal{R}}
\newcommand{\A}{\mat{A}}
\newcommand{\hA}{\hat{\mat{A}}}
\newcommand{\tA}{\tilde{\mat{A}}}
\newcommand{\B}{\mat{B}}
\newcommand{\C}{\mat{C}}

\newcommand{\tC}{\tilde{\mat{C}}}
\newcommand{\azero}{\boldsymbol{\alpha}_0}
\newcommand{\tazero}{\tilde{\boldsymbol{\alpha}}_0}
\newcommand{\hazero}{\hat{\boldsymbol{\alpha}}_0}
\newcommand{\ainf}{\boldsymbol{\alpha}_{\infty}}
\newcommand{\hainf}{\hat{\boldsymbol{\alpha}}_{\infty}}
\newcommand{\tainf}{\tilde{\boldsymbol{\alpha}}_{\infty}}
\newcommand{\bzero}{\boldsymbol{\beta}_0}
\newcommand{\binf}{\boldsymbol{\beta}_{\infty}}
\newcommand{\czero}{\boldsymbol{\gamma}_0}

\newcommand{\cinf}{\boldsymbol{\gamma}_{\infty}}

\newcommand{\tcinf}{\tilde{\boldsymbol{\gamma}}_{\infty}}
\newcommand{\wa}{\langle \azero, \ainf, \{\A_\sigma\} \rangle}

\newcommand{\waS}{\langle \azero, \ainf, \{\A_\sigma\}_{\sigma \in \Sigma} \rangle}
\newcommand{\waQ}{\langle \mQ^\top \azero, \mQ^{-1} \ainf, \{\mQ^{-1} \A_\sigma
\mQ\} \rangle}
\newcommand{\hwa}{\langle \hazero, \hainf, \{\hat{\A}_\sigma\} \rangle}
\newcommand{\twa}{\langle \tazero, \tainf, \{\tilde{\A}_\sigma\} \rangle}

\newcommand{\wb}{\langle \bzero, \binf, \{\B_\sigma\} \rangle}

\newcommand{\wc}{\langle \czero, \cinf, \{\C_\sigma\} \rangle}

\newcommand{\epsopt}{\varepsilon^{\mathsf{opt}}}

\newcommand{\fsva}{\hat{f}^{\mathsf{sva}}}
\newcommand{\fopt}{\hat{f}^{\mathsf{opt}}}

\newcommand{\normop}[1]{\|#1\|_{\mathrm{op}}}
\newcommand{\normsp}[1]{\|#1\|_{\mathrm{S},p}}
\newcommand{\normsone}[1]{\|#1\|_{\mathrm{S},1}}
\newcommand{\normstwo}[1]{\|#1\|_{\mathrm{S},2}}
\newcommand{\normsinf}[1]{\|#1\|_{\mathrm{S},\infty}}
\newcommand{\normhp}[1]{\|#1\|_{\mathrm{H},p}}
\newcommand{\normhone}[1]{\|#1\|_{\mathrm{H},1}}

\newcommand{\normtr}[1]{\|#1\|_{\mathrm{tr}}}
\newcommand{\normf}[1]{\|#1\|_{\mathrm{F}}}

\renewcommand{\H}{\mat{H}}

\renewcommand{\v}{\mat{v}}
\newcommand{\mM}{\mat{M}}
\newcommand{\mN}{\mat{N}}
\newcommand{\mP}{\mat{P}}
\newcommand{\mS}{\mat{S}}
\newcommand{\mG}{\mat{G}}
\newcommand{\mGa}{\boldsymbol{\Gamma}}
\newcommand{\mU}{\mat{U}}
\newcommand{\mV}{\mat{V}}
\newcommand{\mD}{\mat{D}}
\newcommand{\mQ}{\mat{Q}}
\newcommand{\mR}{\mat{R}}
\newcommand{\mI}{\mat{I}}
\newcommand{\mT}{\mat{T}}
\renewcommand{\mp}{\mat{p}}
\newcommand{\me}{\mat{e}}

\newtheorem{theorem}{Theorem}
\newtheorem{proposition}[theorem]{Proposition}
\newtheorem{lemma}[theorem]{Lemma}

\newtheorem{definition}[theorem]{Definition}

\theoremstyle{remark}
\newtheorem{claim}{Claim}
\newtheorem{assumption}{Assumption}

\title{A Canonical Form for Weighted Automata and Applications to Approximate
Minimization}
\author{Borja Balle}
\author{Prakash Panangaden}
\author{Doina Precup}
\affil{School of Computer Science \\ McGill University \\

{\footnotesize \texttt{\{bballe|prakash|dprecup\}@cs.mcgill.ca}}}

\date{April 23, 2015}

\begin{document}

\maketitle

\begin{abstract}

We study the problem of constructing approximations to a weighted automaton.
Weighted finite automata (WFA) are closely related to the theory of rational
series. A rational series is a function from strings to real numbers that can be
computed by a WFA. Among others, this includes probability distributions
generated by hidden Markov models and probabilistic automata. The relationship
between rational series and WFA is analogous to the relationship between regular
languages and ordinary automata. Associated with such rational series are
infinite matrices called Hankel matrices which play a fundamental role in the
theory of minimal WFA.
Our contributions are:
(1) an effective procedure for computing the singular value decomposition
(SVD) of such \emph{infinite} Hankel matrices based on their finite
representation in terms of WFA;
(2) a new \emph{canonical form} for WFA based on this SVD decomposition; and,
(3) an algorithm to construct \emph{approximate minimizations} of a given WFA.
The goal of our approximate minimization algorithm is to start from a minimal
WFA and produce a smaller WFA that is close to the given one in a certain
sense. The desired size of the approximating automaton is given as input. We
give bounds describing how well the approximation emulates the behavior of the
original WFA.

The study of this problem is motivated by the analysis of machine learning
algorithms that synthetize weighted automata from spectral decompositions of
finite Hankel matrices.
It is known that when the number of states of the target automaton is correctly
guessed, these algorithms enjoy consistency and finite-sample guarantees in the
probably approximately correct (PAC) learning model.
It has also been suggested that asking the learning algorithm to produce a
model smaller than the true one will still yield useful models with reduced
complexity.
Our results in this paper vindicate these ideas and confirm intuitions
provided by empirical studies.
%
%
Beyond learning problems, our techniques can also be used to reduce the
complexity of any algorithm working with WFA, at the expense of incurring a
small, controlled amount of error.

\end{abstract}


\section{Introduction}\label{sec:intro}

We address a relatively new issue for the logic and computation community:
the \emph{approximate} minimization of transition systems or automata.
This concept is appropriate for systems that are quantitative in some sense:
weighted automata, probabilistic automata of various kinds and timed
automata.  This paper focusses on weighted automata where we are able to
make a number of contributions that combine ideas from duality with ideas
from the theory of linear operators and their spectrum.  Our new contributions
are
\begin{itemize}
\item An algorithm for the SVD decomposition of \emph{infinite} Hankel matrices
based on their representation in terms of weighted automata.

\item A new canonical form for weighted automata arising from the SVD of its
corresponding Hankel matrix.

\item An algorithm to construct approximate minimizations of given weighted
automata by truncating the canonical form.
\end{itemize}

Minimization of automata has been a major subject since the 1950s, starting
with the now classical work of the pioneers of automata theory.  Recently
there has been activity on novel algorithms for minimization based on
duality~\cite{Bezhanishvili12,Bonchi14} which are ultimately based on a
remarkable algorithm due to Brzozowski from the 1960s~\cite{Brzozowski62}.
The general co-algebraic framework permits one to generalize Brzozowski's
algorithm to other classes of automata like weighted automata.

Weighted automata are very useful in a variety of practical settings, such
as machine learning (where they are used to represent predictive models for
time series data and text), but also in the general theory of quantitative
systems.  There has also been interest in this type of representation, for
example, in concurrency theory~\cite{boreale2009weighted} and in
semantics~\cite{Bonchi2012}. We discuss the machine learning motivations at
greater length, as they are the main driver for the present work. However,
we emphasize that the genesis of one set of key ideas came from previous
work on a coalgebraic view of minimization.


Spectral techniques for learning latent variable models have recently drawn a
lot of attention in the machine learning community.
Following the significant milestone papers~\cite{hsu09,denis}, in which an
efficient spectral 
algorithm for learning hidden Markov models (HMM) and stochastic rational
languages was given, the field has grown very rapidly.
The original algorithm, which is based on singular value decompositions of
Hankel matrices, has been extended to reduced-rank HMM \cite{siddiqi10},
predictive state representations (PSR) \cite{Boots:2009}, finite-state
transducers \cite{fst,bailly2013fst}, and many other classes of functions on
strings \cite{bailly11,nips12,recasens2013spectral}.
Although each of these papers works with slightly different problems and
analyses techniques, the key ingredient turns out to be always the same:
parametrize the target model as a weighted finite automaton (WFA) and learn
this WFA from the SVD of a finite sub-block of its Hankel matrix
\cite{mlj13spectral}.
Therefore, it is possible (and desirable) to study all these learning algorithms
from the point of view of rational series, which are exactly the class of
real-valued functions on strings that can be computed by WFA.
%
%
In addition to their use in spectral learning algorithms, weighted automata are
also commonly used in other areas of pattern recognition for sequences,
including: speech recognition \cite{MohriPereiraRiley2008}, image compression
\cite{AlbertLari2009}, natural language processing
\cite{knight2009applications}, model checking \cite{baier2009model}, and
machine translation \cite{DeGispert2010}.


Part of the appeal of spectral learning techniques comes from their
computational superiority when compared to iterative algorithms like
Expectation--Maximization (EM) \cite{dempster}.
Another very attractive property of spectral methods is the possibility of
proving rigorous statistical guarantees about the learned hypothesis.
For example, under a realizability assumption, these methods are known to be
consistent and amenable to finite-sample analysis in the PAC sense~\cite{hsu09}.
An important detail is that, in addition to realizability, these results work
under the assumption that the user correctly guesses the number of latent
states of the target distribution.
Though this is not a real caveat when it comes to using these algorithms in
practice -- the optimal number of states can be identified using a model
selection procedure \cite{icml2014balle} -- it is one of the barriers in
extending the statistical analysis of spectral methods to the non-realizable
setting.

Tackling the non-realizability question requires, as a special case, dealing
with the situation in which data is generated from a WFA with $n$ states and the
learning algorithm is asked to produce a WFA with $\hat{n} < n$ states.
This case is already a non-trivial problem which -- barring the noisiness
introduced by the use of statistical data instead of the original WFA -- can be
easily interpreted as an approximate minimization of WFA.
From this point of view, the possibility of using spectral learning algorithms
for approximate minimization of a small class of hidden Markov models has
been recently considered in \cite{kulesza2014low}. This paper also presents some
restricted theoretical results bounding the error between the original and
minimized HMM in terms of the total variation distance. Though incomparable to
ours, these bounds are the closest work in the literature to our
approach\footnote{After the submission of this manuscript we became aware of
the concurrent work \cite{kulesza2015low}, where a problem similar to the one
considered here is addressed, albeit different methods are used and the results
are not directly comparable.}.
Another paper on which the issue of approximate minimization of weighted
automata is considered in a tangential manner is \cite{14KW-ICALP}.
In this case the authors again focus on an $\lone$-like accuracy measure to
compare two automata: an original one, and another one obtained by removing
transitions with small weights occurring during an exact minimization procedure.
Though the removal operation is introduced as a means of obtaining a numerically
stable minimization algorithm, the paper also presents some experiments
exploring the effect of removing transitions with larger weights.
With the exception of these timid results, the problem of approximate
minimization remains largely unstudied.
In the present paper we set out to initiate the systematic study of approximate
minimization of WFA.
We believe our results -- beyond their intrinsic automata-theoretic interest --
will also provide tools for addressing important problems in learning theory,
including the robust statistical analysis of spectral learning algorithms.

Let us conclude this introduction by mentioning the potential wide applicability
of our results in the field of algorithms for manipulating, combining, and
operating with quantitative systems. In particular, the possibility of
obtaining reduced-size models incurring a small, controlled amount of error
might provide a principled way for speeding up a number of such algorithms.

The content of the paper is organized as follows.
Section~\ref{sec:background} defines the notation that will be used
throughout the paper and reviews a series of well-known results that will be needed.
Section~\ref{sec:sva} establishes the existence of a canonical form for WFA
and provides a polynomial-time algorithm for computing it (the first major
contribution of this work).
The computation of this canonical form lies at the heart of our approximate
minimization algorithm, which is described and analyzed in
Section~\ref{sec:approxmin}. Our main theoretical result in this section is to
establish bounds describing how well the approximation obtained by the algorithm
emulates the behavior of the original WFA. The proof is quite lengthy and is
deferred to Appendix~\ref{sec:proof}.
In Section~\ref{sec:discuss} we discuss two technical aspects of our work: its
relation and consequences with the mathematical theory of low-rank approximation
of rational series; and the (ir)relevance of an assumption made in our results
from Sections~\ref{sec:sva} and~\ref{sec:approxmin}.
We conclude with Section~\ref{sec:conclusion}, where we point out interesting
future research directions.


%
%
%

%

\section{Background}\label{sec:background}

\subsection{Notation for Matrices}
Given a positive integer $d$, we denote $[d] = \{1,\ldots,d\}$.
We use bold letters to denote vectors $\v \in \R^d$ and matrices $\mM \in
\R^{d_1 \times d_2}$.
Unless explicitly stated, all vectors are column vectors.
We write $\mI$ for the identity matrix, $\diag(a_1,\ldots,a_n)$ for a
diagonal matrix with $a_1, \ldots, a_n$ in the diagonal, and
$\diag(\mM_1,\dots,\mM_n)$ for the block-diagonal matrix containing the square
matrices $\mM_i$ along the diagonal.
The $i$th coordinate vector $(0,\ldots,0,1,0,\ldots,0)^\top$ is denoted by
$\me_i$.
For a matrix $\mM \in \R^{d_1 \times d_2}$, $i \in [d_1]$, and $j \in [d_2]$, we
use $\mM(i,:)$ and $\mM(:,j)$ to denote the $i$th row and the $j$th column of
$\mM$ respectively.
Given a matrix $\mM \in \R^{d_1 \times d_2}$ we can consider the vector
$\vecm(\mM) \in \R^{d_1 \cdot d_2}$ obtained by concatenating the columns of
$\mM$ so that $\vecm(\mM)((i-1) d_2 + j) = \mM(i,j)$.
Given two matrices $\mM \in \R^{d_1 \times d_2}$ and $\mM' \in \R^{d_1' \times
d_2'}$ we denote their Kronecker (or tensor) product by $\mM \otimes \mM' \in
\R^{d_1 d_1' \times d_2 d_2'}$, with entries given by $(\mM \otimes
\mM')((i-1)d_1' + i', (j-1)d_2' + j') = \mM(i,j) \mM'(i',j')$, where $i \in
[d_1]$, $j \in [d_2]$, $i' \in [d_1']$, and $j' \in [d_2']$.
For simplicity, we will sometimes write $\mM^{\otimes 2} = \mM \otimes \mM$, and
similarly for vectors.
A \emph{rank factorization} of a rank $n$ matrix $\mM \in \R^{d_1 \times d_2}$
is an expression of the form $\mM = \mQ \mR$ where $\mQ \in \R^{d_1 \times n}$
and $\mR \in \R^{n \times d_2}$ are full-rank matrices.

Given a matrix $\mM \in \R^{d_1 \times d_2}$ of rank $n$, its \emph{singular
value decomposition} (SVD)\footnote{To be more precise, this is a \emph{reduced}
singular value decomposition, since the inner dimensions of the decomposition
are all equal to the rank. In this paper we shall always use the term SVD to
mean reduced SVD.} is a decomposition of the form $\mM = \mU \mD \mV^\top$ where
$\mU \in \R^{d_1 \times n}$, $\mD \in \R^{n \times n}$, and $\mV \in \R^{d_2
\times n}$ are such that: $\mU^\top \mU = \mV^\top \mV = \mI$, and $\mD =
\diag(\sval_1,\ldots,\sval_n)$ with $\sval_1 \geq \cdots \geq \sval_n > 0$. The
columns of $\mU$ and $\mV$ are called left and right singular vectors
respectively, and the $\sval_i$ are its singular values.
The SVD is unique (up to sign changes in associate singular vectors) whenever
all inequalities between singular values are strict.
%
%
A similar spectral decomposition exists for bounded operators between separable
Hilbert spaces. In particular, for finite-rank bounded operators one can write
the infinite matrix corresponding to the operator in a fixed basis, and recover
a concept of reduced SVD decomposition for such infinite matrices which shares
the same properties described above for finite matrices \cite{zhu1990operator}.



For $1 \leq p \leq \infty$ we will write $\norm{\v}_p$ for the $\lp$ norm of
vector $\v$.
The corresponding \emph{induced norm} on matrices is $\norm{\mM}_p =
\sup_{\norm{\v}_p = 1} \norm{\mM \v}_p$.
In addition to induced norms, we will also need to define Schatten norms.
If $\mM$ is a rank-$n$ matrix with singular values $\mat{s} = (\sval_1, \ldots,
\sval_n)$, the \emph{Schatten $p$-norm} of $\mM$ is given by $\normsp{\mM} =
\norm{\mat{s}}_p$.
Most of these norms have given names: $\norm{\cdot}_2 = \normsinf{\cdot} =
\normop{\cdot}$ is the \emph{operator (or spectral) norm}; $\normstwo{\cdot} =
\normf{\cdot}$ is the \emph{Frobenius norm}; and $\normsone{\cdot} =
\normtr{\cdot}$ is the \emph{trace (or nuclear) norm}.
For a matrix $\mM$ the \emph{spectral radius} is the largest modulus $\rho(\mM)
= \max_i |\lambda_i(\mM)|$ among the eigenvalues of $\mM$.
For a square matrix $\mM$, the series $\sum_{k \geq 0} \mM^k$ converges if and
only if $\rho(\mM) < 1$, in which case the sum yields $(\mI - \mM)^{-1}$.

Sometimes we will name the columns and rows of a matrix using
ordered index sets $\mathcal{I}$ and $\mathcal{J}$. In this case we will
write $\mat{M} \in \R^{\mathcal{I} \times \mathcal{J}}$ to denote a matrix of
size $|\mathcal{I}| \times |\mathcal{J}|$ with rows indexed by $\mathcal{I}$ and
columns indexed by $\mathcal{J}$.

\subsection{Weighted Automata, Rational Series, and Hankel Matrices}

Let $\Sigma$ be a fixed finite alphabet with $|\Sigma| = k$ symbols, and
$\sstar$ the set of all finite strings with symbols in $\Sigma$. We use
$\lambda$ to denote the empty string.
Given two strings $p, s \in \sstar$ we write $w = p s$ for their concatenation,
in which case we say that $p$ is a prefix of $w$ and $s$ is a suffix of $w$.
We denote by $|w|$  the length (number of symbols) in a string $w \in
\sstar$.
Given a set of strings $X \subseteq \sstar$ and a function $f : \sstar \to \R$,
we denote by $f(X)$ the summation $\sum_{x \in X} f(x)$ if defined. For example,
we will write $f(\Sigma^t) = \sum_{|x| = t} f(x)$ for any $t \geq 0$.

Now we introduce our notation for weighted automata. We want to note that we
will not be dealing with weights on arbitrary semi-rings; this paper only
considers automata with real weights, with the usual addition and multiplication
operations. In addition, instead of resorting to the usual description of
automata as directed graphs with labeled nodes and edges, we will use a
linear-algebraic representation, which is more convenient.
A \emph{weighted finite automata} (WFA) of dimension $n$ over $\Sigma$ is a
tuple $A = \waS$ where $\azero \in \R^n$ is the vector of \emph{initial
weights}, $\ainf \in \R^n$ is the vector of \emph{final weights}, and for each
symbol $\sigma \in \Sigma$ the matrix $\A_\sigma \in \R^{n \times n}$ contains
the \emph{transition weights} associated with $\sigma$.
Note that in this representation a fixed initial state is given by $\azero$ (as
opposed to formalisms that only specify a transition structure), and the
transition endophormisms $\A_\sigma$ and the final linear form $\ainf$ are given
in a fixed basis on $\R^n$ (as opposed to abstract descriptions where these
objects are represented as basis-independent elements over some $n$-dimensional
vector space).

We will use $\dim(A)$ to denote the dimension of a WFA. The
state-space of a WFA of dimension $n$ is identified with the integer set $[n]$.
Every WFA $A$ \emph{realizes} a function $f_A : \sstar \to \R$ which, given a
string $x = x_1 \cdots x_t \in \sstar$, produces
\begin{equation*}
f_A(x) = \azero^\top \A_{x_1} \cdots \A_{x_t} \ainf =
\azero^\top \A_x \ainf \enspace,
\end{equation*}
where we defined the shorthand notation $\A_x = \A_{x_1} \cdots \A_{x_t}$ that
will be used throughout the paper.
A function $f : \sstar \to \R$ is called \emph{rational} if there exists a WFA
$A$ such that $f = f_A$.
The \emph{rank} of a rational function $f$ is the dimension of the smallest WFA
realizing $f$.
We say that a WFA is \emph{minimal} if $\dim(A) = \rank(f_A)$.

%

An important operation on WFA is \emph{conjugation} by an invertible matrix.
Suppose $A$ is a WFA of dimension $n$ and $\mQ \in \R^{n \times n}$ is
invertible. Then we can define the WFA
\begin{equation}\label{eqn:wfaconjugation}
A' = \mQ^{-1} A \mQ = \waQ \enspace.
\end{equation}
It is immediate to check that $f_A = f_{A'}$.
This means that the function computed by a WFA is invariant under conjugation,
and that given a rational function $f$, there exist infinitely many WFA realizing
$f$.
In addition, the following result characterizes all minimal WFA realizing a
particular rational function.

\begin{theorem}[\cite{berstel2011noncommutative}]\label{thm:conjugacy}
If $A$ and $B$ are minimal WFA realizing the same function, then $B = \mQ^{-1}
A \mQ$ for some invertible $\mQ$.
\end{theorem}

A function $f : \sstar \to \R$ can be trivially identified with an element from
the free vector space $\R^{\sstar}$.
This vector space contains several subspaces which will play an important role
in the rest of the paper.
One is the subspace of all rational functions, which we denote by $\cR(\Sigma)$.
Note that $\cR(\Sigma)$ is a linear subspace, because if $f, g
\in \cR(\Sigma)$ and $c \in \R$, then $c f$ and $f + g$ are both rational
\cite{berstel2011noncommutative}.
Another important family of subspaces of $\R^{\sstar}$ are the ones containing
all functions with \emph{finite} $p$-norm for some $1 \leq p \leq \infty$, which is
given by $\norm{f}_p^p = \sum_{x \in \sstar} |f(x)|^p$ for finite $p$, and
$\norm{f}_\infty = \sup_{x \in \sstar} |f(x)|$; we denote this space by
$\lp(\Sigma)$.
Note that like in the usual theory of Banach spaces of sequences, we have
$\lp(\Sigma) \subset \lpq(\Sigma)$ for $p < q$.
Of these, $\ltwo(\Sigma)$ can be endowed with the structure of a separable
Hilbert space with the inner product $\langle f, g \rangle = \sum_{x \in \sstar}
f(x) g(x)$. Recall that in this case we have the \emph{Cauchy--Schwarz
inequality} $\langle f, g \rangle^2 \leq \norm{f}_2^2 \, \norm{g}_2^2$.
In addition, we have its generalization, \emph{Hölder's inequality}: given $f
\in \lp(\Sigma)$ and $g \in \lpq(\Sigma)$ with $p^{-1} + q^{-1} \leq 1$, then
$\norm{f \cdot g}_1 \leq \norm{f}_p \norm{g}_q$, where $(f \cdot g)(x) = f(x)
g(x)$.
By intersecting any of the previous subspaces with $\cR(\Sigma)$ one obtains
$\lprat(\Sigma) = \cR(\Sigma) \cap \lp(\Sigma)$, the normed vector space containing
all rational functions with finite $p$-norm.
In most cases the alphabet $\Sigma$ will be clear from the context and we will
just write $\cR$, $\lp$, and $\lprat$.

The space $\lonerat$ of \emph{absolutely convergent rational series} will play a
central role in the theory to be developed in this paper.
An important example of functions in $\lonerat$ is that of probability
distributions over $\sstar$ realized by WFA, also known as rational stochastic
languages.
Formally speaking, these are rational functions $f \in \cR$ satisfying the
constraints $f(x) \geq 0$ and $\sum_x f(x) = 1$.
%
%
This implies that $\lonerat$ includes all functions realized by probabilistic
automata with stopping probabilities \cite{DupontDE05}, hidden Markov models
with absorbing states \cite{Rabiner:1990}, and  predictive state representations
for dynamical systems with discounting or finite horizon \cite{Singh:2004}.
%
Note that given a WFA $A$, the membership problem $f_A \in \lonerat$ is known
to be semi-decidable \cite{decideabsconv}.


Let $\H \in \R^{\sstar \times \sstar}$ be a bi-infinite matrix whose rows and
columns are indexed by strings. We say that $\H$ is \emph{Hankel}%
\footnote{In real analysis a matrix $\mM$ is Hankel if $\mM(i,j)=\mM(k,l)$
whenever $i + j = k + l$, which implies that $\mM$ is symmetric. In our case we
have $\H(p,s) = \H(p',s')$ whenever $p s = p' s'$, but $\H$ is not symmetric
because string concatenation is not commutative whenever $|\Sigma| > 1$.}
if for all strings $p, p', s, s' \in \sstar$ such that $p s = p' s'$ we have
$\H(p,s) = \H(p',s')$.
Given a function $f : \sstar \to \R$ we can associate with it a Hankel matrix
$\H_f \in \R^{\sstar \times \sstar}$ with entries $\H_f(p,s) = f(p s)$.
Conversely, given a matrix $\H \in \R^{\sstar \times \sstar}$ with the Hankel
property, there exists a unique function $f : \sstar \to \R$ such that $\H_f =
\H$.
The following well-known theorem characterizes all Hankel matrices of finite
rank.

\begin{theorem}[\cite{berstel2011noncommutative}]\label{lem:fundamentalWFA}
For any function $f : \sstar \to \R$, the Hankel matrix $\H_f$ has finite rank
$n$ if and only if $f$ is rational with $\rank(f) = n$. In other words,
$\rank(f) = \rank(\H_f)$ for any function $f : \sstar \to \R$.
\end{theorem}

\section{A Canonical Form for WFA}\label{sec:sva}

In this section we discuss the existence and computation of a canonical form for
WFA realizing absolutely convergent rational functions.
Our canonical form is strongly related to the singular value decomposition of
infinite Hankel matrices.
In particular, its existence and uniqueness is a direct consequence of the
existence and uniqueness of SVD for Hankel matrices of functions in $\lonerat$,
as we shall see in the first part of this section.
Furthermore, the algorithm given in Section~\ref{sec:compsva} for computing the
canonical form can also be interpreted as a procedure for computing the SVD of
an infinite Hankel matrix.

\subsection{Existence of the Canonical Form}\label{sec:svaexists}

A matrix $\mT \in \R^{\sstar \times \sstar}$ can be interpreted as the
expression of a (possibly unbounded) linear operator $T : \ltwo \to
\ltwo$ in terms of the canonical basis $\{\me_x\}_{x \in \sstar}$.
In the case of a Hankel matrix $\H_f$, the associated operator $H_f$ is called a
\emph{Hankel operator}, and corresponds to the convolution-like operation $(H_f
g)(x) = \sum_y f(x y) g(y)$ (assuming the series converges).

Recall the \emph{operator norm} of $T : \ltwo \to \ltwo$ is defined as
$\normop{T} = \sup_{\norm{f}_2 \leq 1} \norm{T f}_2$.
An operator is \emph{bounded} if $\normop{T}$ is finite.
Although not all Hankel operators are bounded, next lemma gives a sufficient
condition for $H_f$ to be bounded.

\begin{lemma}\label{lem:boundedhankel}
If $f \in \lone$, then $H_f$ is bounded.
\end{lemma}
\begin{proof}
Let $h(x) = 1 + |x|$ and note that $f \in \lone$ implies $\sup_x |f(x)| (1 +
|x|) < \infty$; i.e.\ $f \cdot h \in \linf$.
Now let $g \in \ltwo$ with $\norm{g}_2 = 1$ and for any $x \in \sstar$ define the
function $f_x(y) = f(x y)$.
Then we have
\begin{align*}
\norm{H_f g}_2^2 &= \sum_x \left(\sum_y f(x y) g(y)\right)^2
= \sum_x \langle f_x, g \rangle^2 \\
&\leq \norm{g}_2^2 \sum_x \norm{f_x}_2^2
= \sum_x \sum_y f(x y)^2 \\
&= \sum_z (1+|z|) f(z)^2
= \sum_z |f(z)| |(1+|z|) f(z)| \\
&\leq \norm{f}_1 \norm{f \cdot h}_{\infty} < \infty \enspace,
\end{align*}
where we used Cauchy--Schwarz inequality, that the number different ways to split a
string $z$ into a prefix and a suffix equals $1 + |z|$, and Hölder's
inequality. This concludes the proof.
\end{proof}

Theorem~\ref{lem:fundamentalWFA} and Lemma~\ref{lem:boundedhankel} imply
that, for any $f \in \lonerat$, the Hankel matrix $\H_f$ represents a bounded
finite-rank linear operator $H_f$ on the Hilbert space $\ltwo$.
Hence, $\H_f$ admits a reduced singular value decomposition $\H_f = \mU \mD
\mV^\top$ where $\mU, \mV \in \R^{\sstar \times n}$ and $\mD \in \R^{n \times
n}$ with $n = \rank(f)$.
The \emph{Hankel singular values} of a rational function $f \in \lonerat$ are
defined as the singular values of the Hankel matrix $\H_f$.
These singular values can be used to define a new set of norms on $\lonerat$:
the \emph{Schatten--Hankel $p$-norm} of $f \in \lonerat$ is given by $\normhp{f} =
\normsp{\H_f} = \norm{(\sval_1,\ldots,\sval_n)}_p$.
It is straightforward to verify that $\normhp{\cdot}$ satisfies the properties
of a norm.


Note an SVD of $\H_f$ yields a rank factorization given by $\H_f = (\mU
\mD^{1/2}) (\mV \mD^{1/2})^\top$.
But SVD is not the only way to obtain rank factorizations for Hankel matrices. In
fact, if $f$ is rational, then every minimal WFA $A$ realizing $f$ induces a
rank factorization of $\H_f$ as follows.
Let $\mP_A \in \R^{\sstar \times n}$ be the \emph{forward matrix} of $A$ given
by $\mP_A(p,:) = \azero^\top \A_p$ for any string $p \in \sstar$.
Similarly, let $\mS_A \in \R^{\sstar \times n}$ be the \emph{backward matrix} of
$A$ given by $\mS_A(s,:) = (\A_s \ainf)^\top$ for any string $s \in \sstar$.
Since $H_f(p,s) = f(p s) = \azero^\top \A_p \A_s \ainf = \mP_A(p,:)
\mS_A^\top(:,s)$, we obtain $\H_f = \mP_A \mS_A^\top$.
This is known as the \emph{forward--backward} (FB) rank factorization of $\H_f$
induced by $A$ \cite{mlj13spectral}.
The following result shows that among the infinity of minimal WFA
realizing a given rational function $f \in \lonerat$, there exists one whose
induced FB rank factorization coincides with $\H_f = (\mU \mD^{1/2}) (\mV
\mD^{1/2})^\top$.

\begin{theorem}\label{thm:sva}
Let $f \in \lonerat$ and suppose $\H_f = (\mU \mD^{1/2}) (\mV \mD^{1/2})^\top$
is a rank factorization induced by SVD.
Then there exists a minimal WFA $A$ for $f$ inducing the same rank
factorization. That is, $A$ induces a FB rank factorization of $\H_f$
given by $\mP_A = \mU \mD^{1/2}$ and $\mS_A = \mV \mD^{1/2}$.
\end{theorem}

Since we have already established the existence of an SVD for $\H_f$ whenever $f
\in \lonerat$, the theorem is just a direct application of the following lemma.

\begin{lemma}\label{lem:duality}
Suppose $f \in \lonerat$ and $\H_f = \mP \mS^\top$ is a rank factorization. Then
there exists a minimal WFA $A$ realizing $f$ which induces this factorization.
\end{lemma}
\begin{proof}
Let $B$ be any minimal WFA realizing $f$ and denote $n = \rank(f)$. Then we have
two rank factorizations $\mP \mS^\top = \mP_B \mS_B^\top$ for the Hankel matrix
$\H_f$.
Therefore, the columns of $\mP$ and $\mP_B$ both span the same $n$-dimensional
sub-space of $\R^{\sstar}$, and there exists a change of basis $\mQ \in \R^{n
\times n}$ such that $\mP_B \mQ = \mP$. This implies we must also have
$\mS^\top = \mQ^{-1} \mS_B^\top$.
It follows that $A = \mQ^{-1} B \mQ$ is a minimal WFA for
$f$ inducing the desired rank factorization.
\end{proof}

The results above leads us to our first contribution: the definition of a
canonical form for WFA realizing functions in $\lonerat$.

\begin{definition}
Let $f \in \lonerat$. A \emph{singular value automaton} (SVA) for $f$ is a
minimal WFA $A$ realizing $f$ such that the FB rank factorization of $\H_f$
induced by $A$ has the form given in Theorem~\ref{thm:sva}.
\end{definition}

Note the SVA provided by Theorem~\ref{thm:sva} is unique up to
the same conditions in which SVD is unique. In particular, it is easy to verify
that if the Hankel singular values of $f \in \lonerat$ satisfy the strict
inequalities $\sval_1 > \cdots > \sval_n$, then the transition weights of the
SVA $A$ of $f$ are uniquely defined, and the initial and final weights are
uniquely defined up to sign changes.

Then next subsection gives a polynomial-time algorithm for computing the SVA of
a function $f \in \lonerat$ starting from a WFA realizing $f$.

\subsection{Computing the Canonical Form}\label{sec:compsva}

As we have seen above, a bi-infinite Hankel matrix $\H_f$ of rank $n$ can
actually be represented with the $n (2 + kn)$ parameters needed to specify the
initial, final and transtion weights of a minimal WFA $A$ realizing $f$.
Though in principle $A$ contains enough information to reconstruct $\H_f$, a
priori it is not clear that $A$ provides an efficient representation for
operating on $\H_f$.
Luckily, it turns out WFA possess a rich algebraic structure allowing
many operations on rational functions and their corresponding Hankel matrices
to be performed in ``compressed'' form by operating directly on WFA
representing them \cite{berstel2011noncommutative}.
In this section we show it is also possible to compute the SVD of $\H_f$ by
operating on a minimal WFA realizing $f$; that is, we give an algorithm for
computing SVA representations.

%
%
%

We start with a simple linear algebra fact showing how to leverage a rank
factorization of a given matrix in order to compute its reduced SVD.
Let $\mM \in \R^{p \times s}$ be a matrix of rank $n$ and suppose $\mM = \mP
\mS^\top$ is a rank factorization.
Let $\mG_p = \mP^\top \mP \in \R^{n \times n}$ be the Gram matrix of the columns
of $\mP$. Since $\mG_p$ is positive definite, it admits a spectral
decomposition $\mG_p = \mV_p \mD_p \mV_p^\top$.
Similarly, we have $\mG_s = \mS^\top \mS = \mV_s \mD_s \mV_s^\top$.
With this notation we have the following.

\begin{lemma}\label{lem:rankfactSVD}
Let $\tilde{\mM} = \mD_p^{1/2} \mV_p^\top \mV_s \mD_s^{1/2}$ with reduced SVD
$\tilde{\mM} = \tilde{\mU} \tilde{\mD} \tilde{\mV}^\top$.
If $\mQ_p = \mV_p \mD_p^{-1/2} \tilde{\mU}$, $\mU = \mP \mQ_p$, $\mQ_s = \mV_s
\mD_s^{-1/2} \tilde{\mV}$, $\mV = \mS \mQ_s$, and $\mD = \tilde{\mD}$, then $\mM
= \mU \mD \mV^\top$ is a reduced SVD for $\mM$.
\end{lemma}
\begin{proof}
We just need to check the columns of $\mU$ and $\mV$ are orthonormal, and
$\mM = \mU \mD \mV^\top$:
\begin{align*}
\mU^\top \mU
&= \mQ_p^\top \mP^\top \mP \mQ_p \\
&= \tilde{\mU}^\top \mD_p^{-1/2} \mV_p^\top \mG_p \mV_p \mD_p^{-1/2} \tilde{\mU} \\
&= \tilde{\mU}^\top \mD_p^{-1/2} \mV_p^\top \mV_p \mD_p \mV_p^\top \mV_p \mD_p^{-1/2} \tilde{\mU} \\
&= \tilde{\mU}^\top \tilde{\mU} \\
&= \mI \enspace, \\
\mV^\top \mV
&= \mQ_s^\top \mS^\top \mS \mQ_s \\
&= \tilde{\mV}^\top \mD_s^{-1/2} \mV_s^\top \mG_s \mV_s \mD_s^{-1/2} \tilde{\mV} \\
&= \tilde{\mV}^\top \mD_s^{-1/2} \mV_s^\top \mV_s \mD_s \mV_s^\top \mV_s \mD_s^{-1/2} \tilde{\mV} \\
&= \tilde{\mV}^\top \tilde{\mV} \\
&= \mI \enspace, \\
\mU \mD \mV^\top
&= \mP \mQ_p \mD \mQ_s^\top \mS^\top \\
&= \mP \mV_p \mD_p^{-1/2} \tilde{\mU} \tilde{\mD} \tilde{\mV}^\top \mD_s^{-1/2} \mV_s^\top \mS^\top \\
&= \mP \mV_p \mD_p^{-1/2} \tilde{\mM} \mD_s^{-1/2} \mV_s^\top \mS^\top \\
&= \mP \mV_p \mD_p^{-1/2} \mD_p^{1/2} \mV_p^\top \mV_s \mD_s^{1/2} \mD_s^{-1/2} \mV_s^\top \mS^\top \\
&= \mP \mS^\top \\
&= \mM \enspace.
\end{align*}
\end{proof}

Note the above result we does not require $p$ and $s$ to be finite.
In particular, when $\mM$ is an infinite matrix associated with a finite-rank
bounded operator, the computation of $\mQ_p$ and $\mQ_s$ can still be done
efficiently as long as $\mG_p$ and $\mG_s$ are available.

Our goal now is to leverage this result in order to compute the SVD of the
bi-infinite Hankel matrix $\H_f$ associated with a rational function $f \in
\lonerat$. The key step will be to compute the Gram matrices associated with
the rank factorization induced by a minimal WFA for $f$.
We start with a lemma showing how to compute the inner product between two
rational functions.

\begin{lemma}\label{lem:innerprod}
Let $A = \wa$ and $B = \wb$ be minimal WFA realizing functions $f_A, f_B \in
\ltworat$.
Suppose the spectral radius of the matrix $\C = \sum_\sigma \A_\sigma \otimes
\B_\sigma$ satisfies $\rho(\C) < 1$.
Then the inner product between $f_A$ and $f_B$ can be computed as:
\begin{equation*}
\langle f_A, f_B \rangle =
\left(\azero \otimes \bzero\right)^\top
\left(\mI - \C \right)^{-1}
\left(\ainf \otimes \binf \right) \enspace.
\end{equation*}
\end{lemma}
\begin{proof}
First note $g(x) = f_A(x) f_B(x)$ is in $\lone$ by Hölder's inequality.
Therefore $\langle f_A, f_B \rangle = \sum_x g(x) \leq \sum_x |g(x)| < \infty$.
In addition, $g$ is rational \cite{berstel2011noncommutative} and can be
computed by the WFA $C = \wc$ given by
\begin{align*}
\czero &= \azero \otimes \bzero \enspace, \\
\cinf &= \ainf \otimes \binf \enspace, \\
\C_\sigma &= \A_\sigma \otimes \B_\sigma \enspace.
\end{align*}
Now note one can use a simple induction argument to show that for any finite $t
\geq 0$ we have
\begin{equation*}\label{eqn:Ck}
s_t = \sum_{x \in \Sigma^t} \czero^\top \C_x \cinf = \czero^\top \C^t \cinf \enspace.
\end{equation*}
Because $g \in \lone$, the series $\sum_{t \geq 0} s_t$ is absolutely
convergent. Thus we must have $\lim_{k \to \infty} \sum_{t \leq k} s_t = L$ for
some finite $L \in \R$.
Since $\rho(\C) < 1$ implies the identity $\sum_{t \geq 0} \C^t = (\mI -
\C)^{-1}$, we must necessarily have $L = \czero^\top (\mI - \C)^{-1} \cinf$.
%
\end{proof}

Note the assumption $\rho(\C) < 1$ is an essential part of our
calculations.
We shall make similar assumptions in the remaining of this
section. See Section~\ref{sec:spectrad} for a discussion on this
assumption and how to remove it.


The following result shows how to efficiently compute the Gram matrices
associated with the rank factorization induced by a minimal WFA for a function
$f \in \lonerat$.

\begin{lemma}\label{lem:computeGpGs}
Let $f \in \lonerat$ with $\rank(f) = n$, and $A = \wa$ be a minimal WFA for
$f$ inducing the FB rank factorization $\H_f = \mP \mS^\top$.
Let $\A^\otimes = \sum_{\sigma} \A_\sigma^{\otimes 2}$.
%
If $\rho(\A^\otimes) < 1$, then
the Gram matrices $\mG_p, \mG_s \in \R^{n \times n}$ associated with the
factorization induced by $A$ satisfy $\vecm(\mG_p)^\top = (\azero^{\otimes
2})^\top (\mI - \A^\otimes)^{-1}$ and $\vecm(\mG_s) = (\mI - \A^\otimes)^{-1}
\ainf^{\otimes 2}$.
\end{lemma}
\begin{proof}
For $i \in [n]$ let $\mp_i = \mP(:,i) \in \R^{\sstar}$ be the $i$th column of
$\mP$.
The key observation is that the function $p_i : \sstar \to \R$ defined by
$p_i(x) = \mp_i(x)$ is in $\ltworat$. To show rationality one just needs to
check $p_i$ is the function realized by the WFA $A_{p,i} = \langle \azero,
\me_i, \{\A_\sigma\} \rangle$ by construction of the induced rank
factorization.
The fact that $\norm{p_i}_2$ is finite follows from Theorem~\ref{thm:sva} by
noting that $\mp_i$ is a linear combination of left singular vectors of
$\H_f$, which belong to $\ltwo$ by definition.
Thus, $\mG_p(i,j) = \mp_i^\top \mp_j$ is well-defined and corresponds to the
inner product $\langle p_i, p_j \rangle$ which, by Lemma~\ref{lem:innerprod}
can be computed as \begin{equation*}
(\azero^{\otimes 2})^\top \left(\mI - \A^\otimes \right)^{-1} (\me_i \otimes \me_j) \enspace.
\end{equation*}
Since $\me_i \otimes \me_j = \me_{(i-1) n + j}$, we obtain the
desired expression for $\vecm(\mG_p)$.
The expression for $\vecm(\mG_s)$ follows from an identical argument using
automata $\A_{s,i} = \langle \me_i, \ainf, \{\A_\sigma\} \rangle$.
\end{proof}

Combining the results above we now show it is possible to compute an
SVA for $f \in \lonerat$ starting from a minimal WFA realizing $f$.
The procedure is called \FuncSty{ComputeSVA} and its description is given in
Algorithm~\ref{alg:sva}.

\begin{algorithm}
\label{alg:sva}
\DontPrintSemicolon
\caption{\FuncSty{ComputeSVA}}
\KwIn{A minimal WFA $A$ with $n$ states for $f \in \lonerat$}
\KwOut{An SVA $A'$ for $f$}
Compute $\mG_p$ and $\mG_s$ \tcc*{cf. Lemma~\ref{lem:computeGpGs}}
Compute $\mQ_p$, $\mQ_s$, and $\mD$ \tcc*{cf.
Lemma~\ref{lem:rankfactSVD}}
Let $A' = \mD^{1/2} \mQ_s^\top A \mQ_p \mD^{1/2}$ \tcc*{cf.
Eq.~(\ref{eqn:wfaconjugation})}
\KwRet{$A'$}\;
\end{algorithm}

\begin{theorem}\label{thm:algosva}
Let $A = \wa$ be a minimal WFA for $f$ such that $\A^\otimes = \sum_\sigma
\A_\sigma^{\otimes 2}$ satisfies $\rho(\A^\otimes) < 1$.
Then the WFA $A'$ computed by $\FuncSty{ComputeSVA}(A)$ is an SVA for $f$.
\end{theorem}
\begin{proof}
Let $\mQ = \mQ_p \mD^{1/2}$. Our first observation is that $\mQ^{-1} =
\mD^{1/2} \mQ_s^\top$ and thus $A$ and $A'$ are minimal WFA for $f$.
Indeed, we already showed in the proof of Lemma~\ref{lem:rankfactSVD} that
\begin{align*}
\mQ_p \mD^{1/2} \mD^{1/2} \mQ_s^\top = \mQ_p \mD \mQ_s^\top = \mI
\enspace.
\end{align*}
In addition, it is immediate to check that if $A$ induces the rank
factorization $\H_f = \mP \mS^\top$, then $A'$ induces the rank factorization
$\H_f = (\mP \mQ_p \mD^{1/2}) (\mD^{1/2} \mQ_s^\top \mS^\top)$, which by
Lemma~\ref{lem:rankfactSVD} satisfies $\mP \mQ_p \mD^{1/2} = \mU \mD^{1/2}$ and
$\mD^{1/2} \mQ_s^\top \mS^\top = \mD^{1/2} \mV^\top$.
\end{proof}

We conclude this section by mentioning that it is possible to modify
$\FuncSty{ComputeSVA}$ to take as input a \emph{non-minimal} WFA $A$ realizing a
function $f \in \lonerat$ under the same assumption on the spectral radius of
the matrix $\A^\otimes$ as we have here.
We shall present the details of this modification somewhere else.
Nonetheless, we note that if given a non-minimal WFA $A$, one always has the
option to minimize $A$ (e.g.\ using the WFA minimization algorithm in
\cite{berstel2011noncommutative}) before attempting the SVA computation.

\subsection{Computational Complexity}

To bound the running time of $\FuncSty{ComputeSVA}(A)$ we recall the following
facts from numerical linear algebra (see e.g.\ \cite{trefethen1997numerical}):
\begin{itemize}
\item The SVD of $\mM \in \R^{d_1 \times d_2}$ ($d_1 \geq d_2$) can be computed
in time $O(d_1 d_2^2)$.
\item The spectral decomposition of a symmetric matrix $\mM \in \R^{d \times d}$
can be computed in time $O(d^3)$.
\item The inverse of an invertible matrix $\mM \in \R^{d \times d}$ can be
computed in time $O(d^3)$.
\end{itemize}

Now note that according to Lemma~\ref{lem:computeGpGs}, computing the Gram
matrices requires $O(k n^4)$ operations to obtain $\mI - \A^\otimes$, plus the
inversion of this $n^2 \times n^2$ matrix, which can be done in time $O(n^6)$.
From Lemma~\ref{lem:rankfactSVD} we see that once the $n \times n$ Gram matrices
$\mG_p$ and $\mG_s$ are given, then computing the singular values $\mD$ and the
change of basis matrices $\mQ_p$ and $\mQ_s$ can be done in time $O(n^3)$.
Finally, the cost of conjugating the WFA $A$ into $A'$ takes time $O(k n^3)$,
where $k = |\Sigma|$ and $n = \dim(A)$.
Hence, the overall running time of $\FuncSty{ComputeSVA}(A)$ is $O(n^6 + k
n^4)$.
Of course, this is a rough estimate which does not take into account
improvements that might be possible in practice, especially in those cases where
the transition matrices of $A$ are sparse -- in such case the complexity of most
operations could be bounded in terms of the number of non-zeros.




\section{Approximate Minimization of WFA}\label{sec:approxmin}

In this section we describe and analyze an approximate minimization algorithm
for WFA. The algorithm takes as input a minimal WFA $A$ with $n$ states and a
target number of states $\hat{n}$, and outputs a new WFA $\hat{A}$ with
$\hat{n}$ states approximating the original WFA $A$.
To obtain $\hat{A}$ we first compute the SVA $A'$ associated to $A$, and then
remove the $n - \hat{n}$ states associated with the smallest singular values of
$\H_{f_A}$.
We call this algorithm \FuncSty{SVATruncation} (see Algorithm~\ref{alg:minsva}
for details).
Since the algorithm only involves a call to \FuncSty{ComputeSVA} and a simple
algebraic manipulation of the resulting WFA, the running time of
\FuncSty{SVATruncation} is polynomial in $|\Sigma|$, $\dim(A)$ and $\hat{n}$.

\begin{algorithm}
\label{alg:minsva}
\DontPrintSemicolon
\caption{\FuncSty{SVATruncation}}
\KwIn{A minimal WFA $A$ with $n$ states, a target number of states $\hat{n} < n$}
\KwOut{A WFA $\hat{A}$ with $\hat{n}$ states}
Let $A' \leftarrow \FuncSty{ComputeSVA}(A)$\;
Let $\mat{\Pi} = [\mat{I}_{\hat{n}} \; \mat{0}] \in \R^{\hat{n} \times n}$\;
Let $\hA_\sigma = \mat{\Pi} \A'_\sigma \mat{\Pi}^\top$ for all $\sigma \in
\Sigma$\;
Let $\hazero = \mat{\Pi} \azero'$\;
Let $\hainf = \mat{\Pi} \ainf'$\;
Let $\hat{A} = \hwa$\;
\KwRet{$\hat{A}$}\;
\end{algorithm}

Roughly speaking, the rationale behind \FuncSty{SVATruncation} is that given an
SVA, the states corresponding to the smallest singular values are the ones with
less influence on the Hankel matrix, and therefore should also be the ones with
less influence on the associated rational function.
However, the details are more tricky than this simple intuition. The reason
being that a low rank approximation to $\H_f$ obtained by truncating its SVD is
not in general a Hankel matrix, and therefore does not correspond to any
rational function. In particular, the Hankel matrix of the function $\hat{f}$
computed by $\hat{A}$ is not obtained by truncating the SVD of $\H_f$.
This makes our analysis non-trivial.

The main result of this section is the following theorem, which bounds the
$\ltwo$-distance between the rational function $f$ realized by the original WFA
$A$, and the rational function $\hat{f}$ realized by the output WFA $\hat{A}$.
The principal attractive of our bound is that it only depends on intrinsic
quantities associated with the function $f$; that is, the final error bound is
independent of which WFA $A$ is given as input.
To comply with the assumptions made in the previous section, we shall assume
like in previous section that the input WFA $A$ satisfies $\rho(\A^{\otimes}) <
1$. The same precepts about this assumption discussed in
Section~\ref{sec:spectrad} apply here.

\begin{theorem}\label{thm:boundsvatrunc}
Let $f \in \lonerat$ with $\rank(f) = n$ and fix $0 < \hat{n} < n$.
If $A$ is a minimal WFA realizing $f$ and such that $\rho(\A^\otimes) < 1$, then
the WFA $\hat{A} = \FuncSty{SVATruncation}(A,\hat{n})$ realizes a function
$\hat{f}$ satisfying
\begin{equation}\label{eq:svatrunc}
\norm{f - \hat{f}}^2_{2} \leq C_f \sqrt{\sval_{\hat{n}+1} + \cdots + \sval_{n}}
\enspace,
\end{equation}
where $C_f$ is a positive constant depending only on $f$.
\end{theorem}

A few remarks about this result are in order.
The first is to observe that because $\sval_1 \geq \cdots \geq \sval_n$, the
error decreases when $\hat{n}$ increases, which is the desired behavior: the
more states $\hat{A}$ has, the closer it is to $A$.
The second is that \eqref{eq:svatrunc} does not depend on which representation
$A$ of $f$ is given as input to \FuncSty{SVATruncation}. This is a consequence of
first obtaining the corresponding SVA $A'$ before truncating.
Obviously, one could obtain another approximate minimization by truncating $A$
directly. However, in that case the final error would depend on the initial $A$
and in general it does not seem possible to use this approach for providing
\emph{representation independent} bounds on the quality of approximation.

The main bulk of the proof of Theorem~\ref{thm:boundsvatrunc} is deferred to
Appendix~\ref{sec:proof}.
Here we will only discuss the basic principle behind the proof and a key
technical lemma which highlights the relevance of the SVA canonical form in our
approach.

The first step in the proof is to combine $A'$ and $\hat{A}$ into a single WFA
$B$ computing $f_B = (f-\hat{f})^2$, and then decompose the error as
\begin{equation*}
\norm{f - \hat{f}}_2^2 = \sum_{t \geq 0} \left(\sum_{x \in \Sigma^t} f_B(x) \right)
\enspace.
\end{equation*}
One can then proceed to bound $f_B(\Sigma^t)$ for all $t \geq 0$ in terms of the
weights of $A'$; this involves lengthy algebraic manipulations with many
intermediate steps exploiting a variety of properties of matrix norms and
Kronecker products.
The last and key step is to exploit the internal structure of the SVA canonical
form in order to turn these bounds into representation independent quantities.
This part of the analysis is based on the following powerful lemma.


\begin{lemma}\label{lem:svaequations}
Let $A = \wa$ be an SVA with $n$ states realizing a function $f \in \lonerat$
with Hankel singular values $\sval_1 \geq \cdots \geq \sval_n$.
Then the following are satisfied:
\begin{enumerate}
\item For all $j \in [n]$, $\sum_{i} \sval_i \sum_{\sigma} \A_{\sigma}(i,j)^2 =
\sval_j - \azero(j)^2$,
\item For all $i \in [n]$, $\sum_{j} \sval_j \sum_{\sigma} \A_{\sigma}(i,j)^2 =
\sval_i - \ainf(i)^2$.
\end{enumerate}
\end{lemma}
\begin{proof}
Recall that $A$ induces the rank factorization $\H_f = \mP \mS^\top = (\mU
\mD^{1/2}) (\mD^{1/2} \mV^\top)$ corresponding to the SVD of $\H_f$.
Let $\mp_j$ be the $j$th column of $\mP = [\mp_1 \cdots \mp_n]$ and note
we have $\norm{\mp_j}_2^2 = \sval_j$.
By appropriately decomposing the sum in $\norm{\mp_j}_2^2$ we get the
following\footnote{Here we are implicitly using the fact that $\sum_{x}
\mp_j(x)^2$ is absolutely (and therefore unconditionally) convergent, which
implies that any rearrangement of its terms will converge to the same value.}:
\begin{equation}
\sval_j = \mp_j(\lambda)^2 + \sum_{\sigma \in \Sigma} \sum_{x \in \sstar}
\mp_j(x \sigma)^2 \enspace. \label{eq:svalj}
\end{equation}
Let us write $\mp_j^\sigma$ for the element of $\ltwo(\Sigma)$ given by
$\mp_j^\sigma(x) = \mp_j(x \sigma)$. Note that by construction we have
$\mp_j^\sigma = \mP \cdot \A_\sigma(:,j) = \sum_{i \in [n]} \mp_i \A_\sigma(i,j)$.
Since $A$ is an SVA, the columns of $\mP$ are orthogonal and therefore we have
\begin{align*}
\norm{\mp_j^\sigma}_2^2 &= \left\langle \sum_i \mp_i \A_\sigma(i,j), \sum_{i'}
\mp_{i'} \A_\sigma(i',j) \right\rangle \\
&= \sum_{i,i'} \A_\sigma(i,j) \A_\sigma(i',j) \langle \mp_i, \mp_{i'} \rangle \\
&= \sum_i \sval_i \A_\sigma(i,j)^2 \enspace.
\end{align*}
Plugging this into \eqref{eq:svalj} and noting that $\mp_j(\lambda) =
\azero(j)$, we obtain the first claim.
The second claim follows from applying the same argument to the columns of
$\mS$.
\end{proof}

To see the importance of this lemma for approximate minimization, let us
consider the following simple consequence which can be derived by combining the
bounds for $\A_\sigma(i,j)$ obtained from considering it belongs to the $i$th row
and the $j$th column of $\A_\sigma$:
\begin{equation*}
|\A_\sigma(i,j)| \leq \min\left\{ \sqrt{\frac{\sval_i}{\sval_j}},
\sqrt{\frac{\sval_j}{\sval_i}} \right\}
= \sqrt{\frac{\min\{\sval_i,\sval_j\}}{\max\{\sval_i,\sval_j\}}} \enspace.
\end{equation*}
This bound is telling us that in an SVA, transition weights further away from
the diagonals of the $\A_\sigma$ are going to be small whenever there is a wide
spread between the largest and smallest singular values; for example,
$|\A_\sigma(1,n)| \leq \sqrt{\sval_n/\sval_1}$.
Intuitively, this means that in an SVA the last states are very weakly connected
to the first states, and therefore removing these connections should not affect
the output of the WFA too much.
Our proof in Appendix~\ref{sec:proof} exploits this intuition and turns it into
a definite quantitative statement.

\section{Technical Discussions}\label{sec:discuss}

This section collects in-detail discussions about two technical aspects of our
work.
The first one is the relation and consequences of our results with respect to
the mathematical theory of low-rank approximation of rational series.
The second part makes some remarks about the assumption on the spectral radius
of WFA made in our results from Sections~\ref{sec:sva} and~\ref{sec:approxmin}.

\subsection{Low-rank Approximation of Rational Series}

We have already discussed how the behavior of the bound \eqref{eq:svatrunc}
matches what intuition suggests. Let us now discuss a little bit more about the
quantitative aspects of the bound.
In particular, we want to make a few observations about the connection of
\eqref{eq:svatrunc} with low-rank approximation of matrices.
We hope these will shed some light on the mathematical theory of low-rank
approximation of rational series, and its relations with low-rank approximations
of infinite Hankel matrices -- a question which certainly deserves further
investigation.


Recall that given a rank-$n$ matrix $\mM \in \R^{d_1 \times d_2}$ and some $1
\leq \hat{n} < n$, the matrix low-rank approximation problem asks for a matrix
$\hat{\mM}$ attaining the optimal of the following optimization
problem:
\begin{equation*}
\min_{\rank(\mM') \leq \hat{n}} \normf{\mM - \mM'}
\enspace.
\end{equation*}
It is well-known the solution to this problem can be computed using the SVD
of $\mM$ and satisfies the following error bounds in terms of Schatten
$p$-norms:
\begin{equation*}
\normsp{\mM - \hat{\mM}} = \norm{(\sval_{\hat{n}+1},\ldots,\sval_n)}_p \enspace.
\end{equation*}

Using these results it is straightforward to give lower bounds on the
approximation errors achievable by low-rank approximation of rational series in
terms of Schatten--Hankel norms (cf.\ Section~\ref{sec:svaexists}).
Let $1 \leq p \leq \infty$ and suppose $f \in \lonerat$ has rank $n$ and Hankel
singular values $\sval_1 \geq \cdots \geq \sval_n$.
Then the following holds for every $f' \in \lonerat$ with $\rank(f') \leq
\hat{n}$:
\begin{equation}\label{eqn:lbound}
\normhp{f - f'} \geq \norm{(\sval_{\hat{n}+1},\ldots,\sval_n)}_p \enspace.
\end{equation}

On the other hand, we define the optimal $\ltwo$ approximation error of $f$ with
respect to all rational functions of rank at most $\hat{n}$ as
\begin{equation*}
\epsopt_{\hat{n}} = \inf_{\rank(f') \leq \hat{n}} \norm{f - f'}_2 \enspace.
\end{equation*}
It is easy to see the infimum will be attained at some $\fopt_{\hat{n}} \in
\ltworat$.
If $\fsva_{\hat{n}} \in \ltworat$ denotes the function realized by the solution
obtained from our $\FuncSty{SVATruncation}$ algorithm, then
Theorem~\ref{thm:boundsvatrunc} implies the bound
\begin{equation}\label{eqn:ubound}
\epsopt_{\hat{n}} \leq \norm{f - \fsva_{\hat{n}}}_2 \leq C_f^{1/2}
\norm{(\sval_{\hat{n}+1},\ldots,\sval_n)}_1^{1/4} \enspace.
\end{equation}

Combining the bounds \eqref{eqn:lbound} and \eqref{eqn:ubound} above, we can
conclude that the performance of our approximation $\fsva_{\hat{n}}$
with respect to $f$ and $\fopt_{\hat{n}}$ can be bounded as follows:
\begin{equation*}
\norm{f - \fopt_{\hat{n}}}_2 \leq \norm{f - \fsva_{\hat{n}}}_2 \leq
C_f^{1/2} \normhone{f - \fopt_{\hat{n}}}^{1/4} \enspace.
\end{equation*}
In future work we plan to investigate the tightness of these bounds and the
computational complexity of (approximately) computing $\fopt_{\hat{n}}$.


\subsection{Spectral Radius Assumptions}\label{sec:spectrad}

The algorithms presented in Sections~\ref{sec:sva} and~\ref{sec:approxmin}
assume their input is a WFA $A = \wa$ such that $\A^\otimes = \sum_\sigma
\A_\sigma \otimes \A_\sigma$ has spectral radius $\rho(\A^\otimes) < 1$.
This condition is used in order to guarantee the existence of a closed-form
expression for the summation of the series $\sum_{t \geq 0} ({\A^{\otimes}})^t$.
Algorithm \FuncSty{ComputeSVA} uses this expression for computing the Gram
matrices $\mG_p = \mP_A^\top \mP_A$ and $\mG_s = \mS_A^\top \mS_A$ associated
with the FB rank factorization $\H_{f_A} = \mP_A \mS_A^\top$ induced by $A$.
A first important remark is that since $\rho(\A^\otimes)$ is defined in terms
of the eigenvalues of $\A^\otimes$, the assumption can be tested efficiently.
The rest of this sections discusses the following two questions: (1) is the
assumption always true in general? and, (2) if not, is there an alternative way
to compute the Gram matrices needed by \FuncSty{ComputeSVA}?

Regarding the first question, let us start by pointing out a natural way in
which one could try to prove that the assumption always hold.
This approach is based on the following result due to F.\ Denis
\cite{denispersonal}.
\begin{proposition}\label{prop:denis}
Let $A = \wa$ be a minimial WFA realizing $f_A \in \lonerat$. Then the spectral
radius of $\A = \sum_\sigma \A_\sigma$ satisfies $\rho(\A) < 1$.
\end{proposition}
In view of this, a natural question to ask is whether the fact $\rho(\A) < 1$
implies $\rho(\A^\otimes) < 1$. While this follows from the equation $\rho(\mM
\otimes \mM) = \rho(\mM)^2$ in the case with $|\Sigma| = 1$, the result is not
true in general for arbitrary matrices. In fact, obtaining interesting bounds on
the spectral radius of matrices of the form $\mM_1 \otimes \mM_1 + \mM_2 \otimes
\mM_2$ is an area of active research in linear algebra
\cite{lototsky2015simple}. Following this approach would require proving new
bounds along these lines that apply to matrices defining WFA for absolutely
convergent rational series.
An alternative approach based on Proposition~\ref{prop:denis} could be to show
that the automaton computing $f_A^2$ obtained in Lemma~\ref{lem:innerprod} is
minimal. However, this is not true in general as witnessed by the following
example.
Let $A$ be the WFA over $\Sigma = \{a,b\}$ with $2$ states given by:
\begin{align*}
\azero^\top &= [1 \;\; 0] \enspace, \\
\ainf^\top &= [1/3 \;\; 1/3] \enspace, \\
\A_a &= \left[ \begin{array}{cc} 0 & 1/3 \\ 1/3 & 0 \end{array} \right] \enspace, \\
\A_b &= \left[ \begin{array}{cc} -1/3 & 0 \\ 0 & 1/3 \end{array} \right]
\enspace. \\
\end{align*}
Note that $\norm{f_A}_1 = 1$ and therefore we have $f_A \in \lonerat$.
It is easy to see, by looking at the rows of $\H_{f_A}$ corrsponding to prefixes
$\lambda$ and $a$, that $\rank(f_A) = 2$; thus, $A$ is minimal.
On the other hand, one can check that $f_A(x)^2 = 3^{-2(|x|+1)}$. Thus, $f_A^2$
has rank $1$ and the $4$-state WFA for $f_A^2$ constructed in
Lemma~\ref{lem:innerprod} is not minimal.
In conclusion, though we have not been able to provide a counter-example to the
fact that $\rho(\A^\otimes) < 1$ when $A$ is a minimal WFA realizing a function
$f_A \in \lonerat$, we suspect that making progress on this problem will require
a deeper understanding of the structure of absolutely convergent rational
series.

The second question is whether it is possible to compute an SVA efficiently for
a WFA such that $\rho(\A^\otimes) \geq 1$.
The key ingredient here is to provide an alternative way of computing the Gram
matrices required in \FuncSty{ComputeSVA}.
A first remark is that such Gram matrices are guaranteed to exist regardless of
whether the assumption on the spectral radius of $\A^\otimes$ holds or not; this
follows from the proof of Lemma~\ref{lem:computeGpGs}.
It also follows from the same proof that each entry $\mG_p(i,j)$ corresponds to
the inner product $\langle p_i, p_j \rangle$ between two rational functions in
$\ltworat$; the same holds for the entries of $\mG_s$.
This observation suggests that, instead of computing the Gram matrices in ``one
shot'' as in Lemma~\ref{lem:computeGpGs}, it might be possible to compute each entry
$\mG_p(i,j)$, $1 \leq i \leq j \leq n$, separately -- note the constraint on the
indices exploits the fact that $\mG_p$ is symmetric.
This can be done as follows. Recall that $A_{p,i} = \langle \azero,
\me_i, \{\A_\sigma\} \rangle$ computes $p_i$ for all $i \in [n]$. Then the
function $f_{p,i,j} = p_i \cdot p_j$ is computed by the WFA $A_{p,i,j} = \langle
\azero^{\otimes 2}, \me_i \otimes \me_j, \{\A_\sigma^{\otimes 2}\} \rangle$ with
$n^2$ states.
Now observe that by Hölder's inequality we have $f_{p,i,j} \in \lonerat$.
Therefore, if $\tilde{\A}_{p,i,j} = \twa$ is a minimization of $A_{p,i,j}$ with
$\rank(f_{p,i,j})$ states, then we must have $\rho(\tilde{\A}) <
1$ by Proposition~\ref{prop:denis}, where $\tilde{\A} = \sum_\sigma
\tilde{\A}_\sigma$. 
Using the same argument as in Lemma~\ref{lem:innerprod} we can conclude that
\begin{align*}
\mG_p(i,j) &= \langle p_i, p_j \rangle \\
&= \sum_{x \in \sstar} f_{p,i,j}(x) \\
&= \tazero^\top (\mI - \tilde{\A})^{-1} \tainf \enspace.
\end{align*}
This gives an alternate procedure for computing $\mG_p$ and $\mG_s$ which involves
$\Theta(n^2)$ WFA minimizations of automata with $n^2$ states, each of them
taking time $O(n^6)$ (cf.\ \cite{berstel2011noncommutative}).
Hence, the cost of this alternate procedure is of order $O(n^8)$, and should only
be used when it is not possible to use the $O(n^6)$ procedure given in
Section~\ref{sec:compsva}.

\section{Conclusions and Future Work}\label{sec:conclusion}

With this paper we initiate a systematic study of approximate minimization
problems of quantitative systems.
Here we have focused our attention on weighted
finite automata realizing absolutely convergent rational series.  These
are, of course, not all rational series but include many situations of
interest, for example, all fully probabilistic automata.
We have shown how the connection between rational series and infinite
Hankel matrices yields powerful tools for analyzing approximate minimization
problems for WFA: the singular value decomposition of infinite Hankel
matrices and the singular values themselves.  Our first contribution: an
algorithm for computing the SVD of an infinite Hankel matrix by operating
on its ``compressed'' representation as a WFA uses these tools in a crucial
way.
Such a decomposition leads us to our second contribution: the
definition of the singular value automaton (SVA) associated with a rational
function $f$.
SVA provide a new canonical form for WFA which is unique under the same
conditions guaranteeing uniqueness of the SVD decomposition for Hankel matrices.
We were also able to give an efficient algorithm for computing the SVA of a
rational function $f$ from any WFA realizing $f$.

Our second set of contributions are related to the application of SVA canonical
forms to the approximate minimization of WFA.
The algorithm \FuncSty{SVATruncation} and the corresponding analysis presented
in Theorem~\ref{thm:boundsvatrunc} provide novel and rigorous bounds on the
quality of our approximations measured in terms of $\norm{f - \hat{f}}_2$,
the $\ltwo$ norm between the original and minimized functions.
The importance of such bounds lies in the fact that they depend only on
intrinsic quantities associated with $f$.

The present paper opens the door to many possible extensions.
First and foremost, we will seek further applications and properties of the SVA
canonical form for WFA. For example, a simple question that remains unanswered is
to what extent the equations in Lemma~\ref{lem:svaequations} are enough to
characterize the weights of an SVA.
In the near future we are also interested in conducting a thorough empirical
study by implementing the algorithms presented here. This should serve to
validate our ideas and explore their possible applications to machine learning
and other applied fields where WFA are used frequently. 
We will also set out to study the tightness of the bound in
Theorem~\ref{thm:boundsvatrunc} in practical situations, and conjecture further
refinements if necessary.
It should also be possible to extend our results to other classes of systems
closely related to weighted automata.  In particular, we want to study
approximate minimization problems for weighted tree automata and weighted
context-free grammars, for which the notions of Hankel matrix can be naturally
extended. Along these lines, it will be interesting to compare our approach to
recent works that improve the running time of parsing algorithms by reducing the
size of probabilistic context-free grammars using low-rank tensor approximations
\cite{collins2012tensor,cohen-13a}.

Though we have not emphasized it in the present paper, this work is
inspired, in part, from the general co-algebraic view of Brzozowski-style
minimization~\cite{Bonchi14}.  We have expressed everything in very
concrete matrix algebra terms because we are using the singular value
decomposition in a crucial way.  However, there are other minimization
schemes for other types of automata coming from other
dualities~\cite{Bezhanishvili12} for which we think similar approximate
minimization schemes can be developed.  A general abstract notion of
approximate minimization is, of course, a very tempting subject to pursue
and, after we have more examples, it would be certainly high on our
agenda.  For the moment, however, we will concentrate on concrete instances.

\bibliographystyle{plain}
\bibliography{paper}

\appendix

\section{Proof of Theorem~\ref{thm:boundsvatrunc}}\label{sec:proof}

Unless stated otherwise, for the purpose of this appendix $\norm{\mat{v}}$
always denotes the Euclidean norm on vectors and $\norm{\mat{M}} =
\sup_{\norm{\mat{v}} = 1} \norm{\mat{M} \mat{v}}$ denotes the corresponding
induced norm on matrices.
We shall use several properties of these norms extensively in the proof.
The reader can consult \cite{bhatia1997matrix} for a comprehensive account of
these properties. In here we will just recall the following:
\begin{enumerate}
\item
$\norm{\mM \mN} \leq \norm{\mM} \norm{\mN}$,
\item
$\norm{\mU \mM \mU^\top} = \norm{\mM}$ whenever $\mU \mU^\top = \mU^\top \mU =
\mI$,
\item
$\norm{\diag(\mM,\mN)} = \max \{ \norm{\mM}, \norm{\mN} \}$,
\item $\norm{\mM} \leq \norm{\mM}_F = \sqrt{\sum_{i,j} \mM(i,j)^2}$,
\item $\norm{\mM \otimes \mN} = \norm{\mM} \norm{\mN}$.
\end{enumerate}

We start by noting that, without loss of generality, we can assume the automaton
$A$ given as input to \FuncSty{SVATruncation} is in SVA form (in which case $A'
= A$).

Now we introduce some notation by splitting the weights conforming $A$ into a
block corresponding to states $1$ to $\hat{n}$, and another block containing
states $\hat{n}+1$ to $n$. With this, we write the following:
\begin{align*}
\azero &= \left[ \azero^{(1)} \; \azero^{(2)} \right] \enspace, \\
\ainf &= \left[ \begin{array}{c} \ainf^{(1)} \\ \ainf^{(2)} \end{array} \right]
\enspace, \\
\A_\sigma &= \left[ \begin{array}{cc} \A_\sigma^{(11)} & \A_\sigma^{(12)} \\
\A_\sigma^{(21)} & \A_\sigma^{(22)} \end{array} \right] \enspace.
\end{align*}
It is immediate to check that $\hat{A} = \FuncSty{SVATruncation}(A,\hat{n})$ is
given by $\hazero = \azero^{(1)}$, $\hainf = \ainf^{(1)}$, and $\hA_\sigma =
\A_\sigma^{(11)}$.
For simplicity of notation, we assume here a throughout the rest of the proof
that initial weights of WFA are given as \emph{row} vectors.

Although $\hat{A}$ has $\hat{n}$ states, it will be convenient for our proof to
find another WFA with $n$ states computing the same function as $\hat{A}$. We
call this WFA $\tilde{A}$, and its construction is explained in the following
claim, whose proof is just a simple calculation.
\begin{claim}
Let $\tilde{A} = \twa$ be the WFA with $n$ states given by $\tazero = \azero$,
\begin{align*}
\tainf &= \left[ \begin{array}{c} \ainf^{(1)} \\ \mat{0} \end{array} \right]
\enspace, \\
\tA_\sigma &= \left[ \begin{array}{cc} \A_\sigma^{(11)} & \mat{0} \\ \mat{0} &
\A_\sigma^{(22)} \end{array} \right] \enspace.
\end{align*}
Then $\tilde{f} = f_{\tilde{A}} = \hat{f}$.
\end{claim}

By combining $A$ and $\tilde{A}$ we can obtain a WFA computing squares of
differences between $f$ and $\hat{f}$. The construction is given in the
following claim, which follows from the same argument used in the proof of
Lemma~\ref{lem:innerprod}.
\begin{claim}
Let $B = {\langle \bzero^{\otimes 2}, \binf^{\otimes 2}, \{\B_\sigma^{\otimes 2}\}
\rangle}$ be the WFA with $4 n^2$ states constructed from
\begin{align*}
\bzero &= \left[ \azero \; -\tazero \right] \enspace,  \\
\binf &= \left[ \begin{array}{c} \ainf \\ \tainf \end{array} \right] \enspace,
\\
\B_\sigma &=  \left[ \begin{array}{cc} \A_\sigma & \mat{0} \\ \mat{0} &
\tA_\sigma \end{array} \right] = \diag(\A_\sigma,\tA_\sigma) \enspace.
\end{align*}
Then $f_B = (f - \tilde{f})^2$.
\end{claim}

From the weights of automaton $B$ we define the following vectors and matrices:
\begin{align*}
\czero &= \azero \otimes [\azero \; -\tazero ] = \azero \otimes \bzero \enspace, \\
\cinf &= \ainf \otimes \left[ \begin{array}{c} \ainf \\ \tainf \end{array}
\right] = \ainf \otimes \binf \enspace, \\
\tcinf &= \tainf \otimes \left[ \begin{array}{c} \ainf \\ \tainf \end{array}
\right] = \tainf \otimes \binf \enspace, \\
\C_\sigma &= \A_\sigma \otimes \left[ \begin{array}{cc} \A_\sigma & \mat{0}
\enspace, \\
\mat{0} & \tA_\sigma \end{array} \right] = \A_\sigma \otimes \B_\sigma \enspace, \\
\tC_\sigma &= \tA_\sigma \otimes \left[ \begin{array}{cc} \A_\sigma & \mat{0}
\\ \mat{0} & \tA_\sigma \end{array} \right] = \tA_\sigma \otimes \B_\sigma
\enspace.
\end{align*}
We will also write $\C = \sum_\sigma \C_\sigma$ and $\tC = \sum_\sigma
\tC_\sigma$.
This notation lets us state the following claim, which will be the starting
point of our bounds.
The result follows from the same calculations used in the proof of
Lemma~\ref{lem:innerprod} and the observation that $\bzero^{\otimes 2} = [\czero
\; -\czero]$.

\begin{claim}
For any $t \geq 0$ we have
\begin{equation*}
\Delta_t  = \sum_{x \in \Sigma^t} (f(x) - \hat{f}(x))^2 = \czero \left( \C^t \cinf - \tC^t
\tcinf \right) \enspace.
\end{equation*}
\end{claim}

Note that we can write the error we want to bound as $\norm{f - \hat{f}}_2^2 =
\sum_{t \geq 0} \Delta_t$. Our strategy will be to obtain a separate bound for
each $\Delta_t$ and then sum them all together.
We start by bounding $|\Delta_t|$ in terms of the norms of the matrices and
vectors defined above.

\begin{lemma}\label{lem:bounddeltat}
For any $t \geq 0$ the following bound holds:
\begin{align*}
|\Delta_t| &\leq
\norm{\czero} \norm{\tC}^t \norm{\cinf - \tcinf} \\
&+ t \norm{\czero} \norm{\cinf} \max\{\norm{\C},\norm{\tC}\}^{t-1} \norm{\C -
\tC} \enspace.
\end{align*}
\end{lemma}
\begin{proof}
Let us start by remarking that the proof of the bound does not depend in any way
on the identity of the vectors and matrices involved. Thus, we shall write
$\Delta_t = \Delta_t(\cinf,\tcinf)$, which will allow us later in the proof to
change the identity of the vectors $\cinf$ and $\tcinf$ appearing in the bound.

Now we proceed by induction on $t$. The base case $t = 0$ is immediate since
\begin{equation*}
|\Delta_0| = |\czero (\cinf - \tcinf)| \leq \norm{\czero} \norm{\cinf
- \tcinf} \enspace.
\end{equation*}
Now assume the bound is true for $\Delta_t$. For the case $t+1$ repeated
applications of the triangle inequality yield:
\begin{align*}
|\Delta_{t+1}|
&=
| \czero ( \C^{t+1} \cinf - \tC^{t+1} \tcinf )| \\
&\leq
|\czero \tC^{t+1} (\cinf - \tcinf)|
\\ &+
|\czero (\C^{t+1} - \tC^{t+1}) \cinf|
\\ &\leq 
|\czero \tC^{t+1} (\cinf - \tcinf)|
\\ &+
|\czero \C^t (\C - \tC) \cinf|
\\ &+
|\czero (\C^t - \tC^t) \tC \cinf|
\\ &\leq
\norm{\czero} \norm{\tC}^{t+1} \norm{\cinf - \tcinf}
\\ &+
\norm{\czero} \norm{\C}^t \norm{\C - \tC} \norm{\cinf}
\\ &+
|\Delta_t(\tC \cinf, \tC \cinf)| \enspace.
\refstepcounter{equation}\tag{\theequation}\label{eqn:deltat1}
\end{align*}
Let $\Delta'_t = \Delta_t(\tC \cinf, \tC \cinf)$. By our inductive hypothesis we
have
\begin{align*}
|\Delta'_t|
&\leq
\norm{\czero} \norm{\tC}^t \norm{\tC \cinf - \tC \cinf}
\\ &+
t \norm{\czero} \norm{\tC \cinf} \max\{\norm{\C},\norm{\tC}\}^{t-1}
\norm{\C - \tC}
\\
&\leq
t \norm{\czero} \norm{\cinf} \max\{\norm{\C},\norm{\tC}\}^{t}
\norm{\C - \tC}
\enspace.
\end{align*}
Finally, plugging this bound into \eqref{eqn:deltat1} we get
\begin{align*}
|\Delta_{t+1}|
&\leq
\norm{\czero} \norm{\tC}^{t+1} \norm{\cinf - \tcinf}
\\ &+
\norm{\czero} \norm{\C}^t \norm{\C - \tC} \norm{\cinf}
\\ &+
t \norm{\czero} \norm{\cinf} \max\{\norm{\C},\norm{\tC}\}^{t}
\norm{\C - \tC}
\\ &=
\norm{\czero} \norm{\tC}^{t+1} \norm{\cinf - \tcinf}
\\ &+
(t+1) \norm{\czero} \norm{\cinf} \max\{\norm{\C},\norm{\tC}\}^{t}
\norm{\C - \tC}
\enspace.
\end{align*}
\end{proof}

Now we proceed to bound all the terms that appear in this bound.
First note that the bounds in the following claim are clear by definition.

\begin{claim}
We have $\norm{\czero} = \sqrt{2} \norm{\azero}^2$ and $\norm{\cinf} \leq
\sqrt{2} \norm{\ainf}^2$.
\end{claim}

The next step is to bound the term involving the norms $\norm{\C}$ and
$\norm{\tC}$. This will lead to the definition of a representation independent
parameter we call $\rho_f$.

\begin{lemma}
Let $\rho_f = \norm{\sum_\sigma \A_\sigma \otimes \A_\sigma}$. Then $\rho_f$ is
a positive constant depending only on $f$ which satisfies:
\begin{equation*}
\norm{\tC} \leq \norm{\C} = \rho_f \enspace.
\end{equation*}
\end{lemma}
\begin{proof}
We start by noting that $\norm{\C} = \max\{\norm{\sum_\sigma \A_\sigma \otimes
\A_\sigma}, \norm{\sum_\sigma \A_\sigma \otimes \tA_\sigma}\}$.
Then we use $\norm{\sum_\sigma \A_\sigma \otimes \tA_\sigma} = \norm{\sum_\sigma
\tA_\sigma \otimes \A_\sigma} = \max\{ \norm{\sum_\sigma \A_\sigma^{(11)}
\otimes \A_\sigma}, \norm{\sum_\sigma \A_\sigma^{(22)}
\otimes \A_\sigma} \}$ to show that $\norm{\C} = \norm{\sum_\sigma \A_\sigma
\otimes \A_\sigma}$, since the rest of terms in the maximum correspond to norms
of submatrices of $\sum_\sigma \A_\sigma \otimes \A_\sigma$.

Now a similar argument can be used to show that
\begin{align*}
\norm{\tC} &= \max\{
\norm{\sum_\sigma \A_\sigma^{(11)} \otimes \A_\sigma},
\norm{\sum_\sigma \A_\sigma^{(11)} \otimes \tA_\sigma},
\\ &
\norm{\sum_\sigma \A_\sigma^{(22)} \otimes \A_\sigma},
\norm{\sum_\sigma \A_\sigma^{(22)} \otimes \tA_\sigma}
\}
\\ &\leq
\norm{\sum_\sigma \A_\sigma \otimes \A_\sigma} \enspace.
\end{align*}

Note $\rho_f$ is representation independent because it only depends on the
transition weights of the SVA form of $f$, which is unique.

%
\end{proof}

In the remaining of the proof we will assume the following holds.
\begin{assumption}\label{ass:rhof}
The SVA $A$ is such that $\rho_f = \norm{\sum_\sigma \A_\sigma^{\otimes 2}} <
1$.
\end{assumption}
This assumption is not essential, and is only introduced to simplify the
calculations involved in the proof. See the remarks at the end of this appendix
for a discussion on how to remove the assumption.

In order to bound $\norm{\cinf - \tcinf}$ and $\norm{\C - \tC}$ we will make an
extensive use of the properties of SVA given in Lemma~\ref{lem:svaequations}.
This will allow us to obtain bounds that only depend on the Hankel singular
values of $f$, which are intrinsic representation-independent quantities
associated with $f$.

\begin{lemma}
\begin{equation*}
\norm{\cinf - \tcinf} \leq \sqrt{2} \norm{\ainf} \sqrt{\sval_{\hat{n}+1} +
\cdots + \sval_{n}} \enspace.
\end{equation*}
\end{lemma}
\begin{proof}
We start by unwinding the definitions of $\cinf$ and $\tcinf$ to obtain the
bound:
\begin{align*}
\norm{\cinf - \tcinf}
&= \norm{(\ainf - \tainf) \otimes \binf}
\\ &=
\norm{\ainf - \tainf} \norm{\binf}
\\ &=
\norm{\ainf^{(2)}} \sqrt{\norm{\ainf}^2 + \norm{\tainf}^2}
\\ &\leq
\sqrt{2} \norm{\ainf} \norm{\ainf^{(2)}} \enspace,
\end{align*}
where we used that $\norm{\tainf} \leq \norm{\ainf}$.
Now note that Lemma~\ref{lem:svaequations} yields the crude estimate $\ainf(i)^2
\leq \sval_i$ for all $i \in [n]$. We use this last observation to obtain the
following bound and complete the proof:
\begin{equation*}
\norm{\ainf^{(2)}} = \sqrt{\sum_{i=\hat{n}+1}^n \ainf(i)^2} \leq
\sqrt{\sval_{\hat{n}+1} + \cdots + \sval_n} \enspace.
\end{equation*}
\end{proof}

\begin{lemma}
\begin{equation*}
\norm{\C - \tC} \leq \sqrt{\sum_\sigma \norm{\A_\sigma}^2}
\sqrt{\frac{\sval_{\hat{n}+1} + \cdots + \sval_{n}}{\sval_{\hat{n}}}} \enspace.
\end{equation*}
\end{lemma}
\begin{proof}
Let $\mGa = \C - \tC = \sum_\sigma (\A_\sigma - \tA_\sigma) \otimes \B_\sigma$.
By expanding $\A_\sigma - \tA_\sigma$ in this expression one can see that
\begin{equation*}
\mGa = \left[ \begin{array}{cc} \mat{0} & \mGa^{(12)} \\ \mGa^{(21)} & \mat{0}
\end{array} \right] \enspace,
\end{equation*}
where $\mGa^{(ij)} = \sum_\sigma \A_\sigma^{(ij)} \otimes \B_\sigma$ for $ij \in
\{12, 21\}$. Since both $\mGa^{(12)}$ and $\mGa^{(21)}$ are unitarily equivalent
to block-diagonals matrices, $\mGa$ is also unitarily equivalent to a
block-diagonal matrix. Thus, we have
\begin{align*}
\norm{\mGa} &= \max \left\{
\norm{\sum_\sigma \A_\sigma^{(12)} \otimes \A_\sigma},
\norm{\sum_\sigma \A_\sigma^{(12)} \otimes \A_\sigma^{(11)}}, \right.
\\ &
\norm{\sum_\sigma \A_\sigma^{(12)} \otimes \A_\sigma^{(22)}},
\norm{\sum_\sigma \A_\sigma^{(21)} \otimes \A_\sigma},
\\ &
\left. \norm{\sum_\sigma \A_\sigma^{(21)} \otimes \A_\sigma^{(11)}},
\norm{\sum_\sigma \A_\sigma^{(21)} \otimes \A_\sigma^{(22)}} \right\} \\
&= \max \left\{
\norm{\sum_\sigma \A_\sigma^{(12)} \otimes \A_\sigma},
\norm{\sum_\sigma \A_\sigma^{(21)} \otimes \A_\sigma} \right\} \enspace,
\end{align*}
where the last equation follows because we just removed from the maximum terms
corresponding to norms of submatrices of the remaining terms. 
Now we can use Lemma~\ref{lem:svaequations} to bound the terms in this
expression (we only give one of the bounds in detail, the other one follows
mutatis mutandis from a mimetic argument).
Let us start with the following simple observation:
\begin{align*}
\norm{\sum_\sigma \A_\sigma^{(12)} \otimes \A_\sigma} &\leq
\sum_\sigma \norm{\A_\sigma^{(12)}} \norm{\A_\sigma} \\ &\leq
\sqrt{\sum_\sigma \norm{\A_\sigma^{(12)}}^2} \sqrt{\sum_\sigma
\norm{\A_\sigma}^2} \enspace.
\end{align*}
The following chain of inequalities provides the final bound:
\begin{align*}
\sum_\sigma \norm{\A_\sigma^{(12)}}^2
&\leq
\sum_\sigma \norm{\A_\sigma^{(12)}}_F^2
\\ &=
\sum_\sigma \sum_{i=1}^{\hat{n}} \sum_{j=\hat{n}+1}^n \A_\sigma(i,j)^2
\\ &=
\sum_{j=\hat{n}+1}^n \sum_{i=1}^{\hat{n}} \frac{\sval_i}{\sval_i} \sum_\sigma
\A_\sigma(i,j)^2
\\ &\leq
\frac{1}{\sval_{\hat{n}}} 
\sum_{j=\hat{n}+1}^n \sum_{i=1}^{n} \sval_i \sum_\sigma \A_\sigma(i,j)^2
\\ &=
\frac{1}{\sval_{\hat{n}}} 
\sum_{j=\hat{n}+1}^n (\sval_j - \azero(j)^2)
\\ &\leq
\frac{\sval_{\hat{n}+1} + \cdots + \sval_n}{\sval_{\hat{n}}} \enspace.
\end{align*}
\end{proof}

Now we can put all the above together in order to obtain the bound in
Theorem~\ref{thm:boundsvatrunc}.
We start with the following bound for $|\Delta_t|$ for some fixed $t$, which
follows from simple algebraic manipulations:
\begin{align*}
|\Delta_t| &\leq \left(C'_1 \rho_f^t + C'_2 \frac{t
\rho_f^{t-1}}{\sqrt{\sval_{\hat{n}}}} \right) \sqrt{\sval_{\hat{n}+1} + \cdots +
\sval_n} \enspace,
\end{align*}
where $C'_1 = 2 \norm{\azero}^2 \norm{\ainf}$ and $C'_2 = 2 \norm{\azero}^2
\norm{\ainf}^2 \left(\sum_\sigma \norm{\A_\sigma}^2\right)^{1/2}$ are constants
depending only on $f$.

Finally, using that $\norm{f - \hat{f}}_2^2 =\sum_{t \geq 0} \Delta_t$, we can
sum all the bounds on $\Delta_t$ and obtain
\begin{align*}
\norm{f - \hat{f}}_2^2 &\leq \left( C_1 + \frac{C_2}{\sqrt{\sval_{\hat{n}}}}
\right) \sqrt{\sval_{\hat{n}+1} + \cdots + \sval_{n}} \\
&\leq C_f \sqrt{\sval_{\hat{n}+1} + \cdots + \sval_{n}}
\enspace,
\end{align*}
with $C_1 = C'_1 / (1 - \rho_f)$, $C_2 = C'_2 / (1 - \rho_f)^2$,
and $C_f = C_1 + C_2/\sqrt{\sval_n}$. 
This concludes the proof of Theorem~\ref{thm:boundsvatrunc}.

Note that Assumption~\ref{ass:rhof} played a crucial role in asserting the
convergence of $\sum_{t \geq 0} \rho_f^t$.
Although the assumption may not hold in general, it can be shown using the
results from \cite[Section 2.3.4]{baillythesis} that there exists a minimal WFA
$D$ with $f_D = f^2$ such that $\norm{\sum_\sigma \mD_\sigma} < 1$.
Thus, by Theorem~\ref{thm:conjugacy} we have $\mD_\sigma = \mQ^{-1}
\A_\sigma^{\otimes 2} \mQ$ for some $\mQ$.
It is then possible to rework the proof of Lemma~\ref{lem:bounddeltat} using $D$
and obtain a similar bound involving $\norm{\sum_\sigma \mD_\sigma}$ instead of
$\norm{\C}$ and $\norm{\tC}$.
The details of this approach will be developed in a longer version of this
paper, but it suffices to say here that in the end one obtains the same bound
given in Theorem~\ref{thm:boundsvatrunc} with a different constant $C'_f$.

\end{document}